\documentclass{article}
\usepackage{ijcai25}

\usepackage{times}
\usepackage{soul}
\usepackage{url}
\usepackage[hidelinks]{hyperref}
\usepackage[utf8]{inputenc}
\usepackage[small]{caption}
\usepackage{graphicx}
\usepackage{amsmath}
\usepackage{amsthm}
\usepackage{booktabs}
\usepackage{algorithm}
\usepackage[switch]{lineno}
\usepackage{algpseudocode}

\usepackage{amssymb}
\usepackage{macros}
\usepackage{mathrsfs}
\usepackage{tikz}
\usepackage{bussproofs}
\usetikzlibrary{shapes.misc, fit, decorations.pathreplacing,calligraphy,positioning} 

\usepackage{soul}
\usepackage{todonotes}

 \algtext*{EndIf}
\algtext*{EndFor}

\algnewcommand\algorithmiccase{\textbf{case}}
\algdef{SE}[CASE]{Case}{EndCase}[1]{\algorithmiccase\ #1}{\algorithmicend\ \algorithmiccase}%
\algtext*{EndCase}

\newtheorem{example}{Example}
\newtheorem{theorem}{Theorem}
\newtheorem{definition}{Definition}
\newtheorem{proposition}{Proposition}
\newtheorem{corollary}{Corollary}
\newtheorem{remark}{Remark}
\newtheorem{lemma}{Lemma}

\title{First-Order Coalition Logic\thanks{This is an extended version of the paper with the same title that appears in the proceedings of IJCAI
2025. This version contains a technical appendix with proof details that, for space reasons, do not appear in the IJCAI 2025 version.}}

\author{
Davide Catta$^1$
\and
Rustam Galimullin$^2$\And
Aniello Murano$^{3}$\\
\affiliations
$^1$LIPN, CNRS UMR 7030, Université Sorbonne Paris Nord,
  Villetaneuse, France\\
$^2$University of Bergen, Norway\\
$^3$University of Naples Federico II, Italy\\
\emails
catta@lipn.univ-paris13.fr, rustam.galimullin@uib.no, 
aniello.murano@unina.it 
}

\begin{document}

\maketitle

\begin{abstract}
We introduce \textit{First-Order Coalition Logic} ($\CSL$), which combines key intuitions behind Coalition Logic ($\mathsf{CL}$) and Strategy Logic ($\mathsf{SL}$). Specifically, $\CSL$ allows for arbitrary quantification over actions of 
agents.  
$\CSL$ is interesting for several reasons. First, we show that $\CSL$ is strictly more expressive than existing coalition logics. 
Second, we provide a sound and complete axiomatisation of $\CSL$, which, to the best of our knowledge, is 
\textit{the first axiomatisation} of any variant of $\mathsf{SL}$ in the literature. 
Finally, while discussing the satisfiability problem for $\CSL$, we reopen the question of the recursive axiomatisability of $\mathsf{SL}$.
\end{abstract}

\section{Introduction}
\textit{Logics for strategic reasoning} constitute a numerous family of formal tools devised to model, verify, and reason about the abilities and strategies of (groups of) autonomous agents in a competitive environment~\cite{pauly02,alur02,van2005logic,mogavero14,chatterjee2010strategy}. Strategies here are ‘recipes’ telling agents what to do in order to achieve their goals. The competitive environment part arises from the fact that in the presence of several agents trying to achieve their own goals, the actions of one agent may influence the available strategies of another agent. Such logics have been shown to be invaluable for specification and verification within various domains: neuro-symbolic reasoning \cite{akintunde20}, voting protocols \cite{jamroga18}, autonomous submarines \cite{ezekiel11}, manufacturing robots \cite{desilva17}, and so on. 

The prime representatives of logics for strategic reasoning are \textit{coalition logic} ($\mathsf{CL}$) \cite{pauly02}, \textit{alternating-time temporal logic} ($\mathsf{ATL}$), \cite{alur02}, and \textit{strategy logic} ($\mathsf{SL}$) \cite{mogavero10} (and numerous variations thereof). $\mathsf{CL}$ extends the language of propositional logic with constructs $\langle \! \langle C \rangle \! \rangle \varphi$ meaning `coalition $C$ has a joint action such that $\varphi$ holds in the next state (no matter what agents outside of the coalition do at the same time)'. $\mathsf{ATL}$ extends further the abilities of agents to force temporal goals expressed with the help of such modalities as `\textsf{U}ntil' and `\textsf{R}elease'. Finally, $\mathsf{SL}$ allows for a more fine-tuned quantification over agents' abilities: while in both $\mathsf{CL}$ and $\mathsf{ATL}$ we have a fixed quantification prefix $\exists \forall$, in $\mathsf{SL}$ we can have arbitrary quantification prefixes. Thus, in $\mathsf{SL}$ we can reason, for example, about agents sharing their strategies, and such game-theoretic notions like dominant strategies and Nash equilibria.  Hence, $\mathsf{ATL}$ is strictly more expressive than $\mathsf{CL}$, and, in turn, $\mathsf{SL}$ is strictly more expressive than $\mathsf{ATL}$ (and its more general cousin $\mathsf{ATL}^\ast$).



Sound and complete axiomatisations of $\mathsf{CL}$ \cite{pauly02,goranko13} and $\mathsf{ATL}$ \cite{goranko06} are now classic results in the field. However, to the best of our knowledge, \textit{no axiomatisations of $\mathsf{SL}$, nor any of its variants, have been considered in the literature so far}. 

In this paper, we introduce a novel variation of the next-time fragment of $\mathsf{SL}$ that we call \textit{first-order coalition logic} ($\CSL$)\footnote{Not to be confused with \textit{quantified coalition logic} \cite{agotnes08}, where quantification is over coalitions and which is as expressive as $\mathsf{CL}$.}. As its name suggests, $\CSL$ combines the coalition reasoning capabilities of $\mathsf{CL}$ with the first-order features of $\mathsf{SL}$. 
Specifically, we allow arbitrary quantification prefixes over agents' actions and also allow action labels to appear explicitly in the language. This makes $\CSL$ to be closely related to $\mathsf{ATL}$ \textit{with explicit strategies} \cite{walther07}. 

We first show that 
$\CSL$ is  quite a special $\mathsf{CL}$, being strictly more expressive than other known coalition logics. With such a remarkable expressivity comes the \textit{PSPACE}-complete model checking problem and the undecidable satisfiability problem. While proving the undecidability result, we have also reopened the problem of the recursive axiomatisability of $\mathsf{SL}$, which was until now assumed to be not recursively axiomatisable \cite{mogavero10}. Moreover, we provide a sound and complete axiomatisation of $\CSL$, which is, as far as we can tell, \textit{the first axiomatisation of any variant of $\mathsf{SL}$}. 
Thus,
we lay the groundwork for the  axiomatisations of more expressive fragments and variants of $\mathsf{SL}$.  

The rest of the paper is structured as follows: Section \ref{sec:csl} defines the syntax and semantics of $\CSL$, Section \ref{sec:expressivity} examines its expressiveness, Section \ref{sec:axiom} presents a complete axiomatisation of $\CSL$, Section \ref{sec:mc} addresses complexity, and Section \ref{sec:conclusion} concludes with directions for future work.
%

\section{Syntax and Semantics}


\label{sec:csl}
\begin{definition}[Language] 

A \emph{signature} is a triple $\alpha= \tuple{n,\mathcal{C},\Ap}$, where $n\geq 1$ is a natural number,  $\mathcal{C}$ is a non-empty  countable 
set of \emph{constants}, and $\Ap$ is a non-empty countable set of \emph{atomic propositions} (or \emph{atoms}) 
such that $\Ap \cap \mathcal{C} = \emptyset$.

Fix 
a non-empty countable set  $\V$   
of \emph{variables} that is disjoint from any other set in any given signature $\alpha$. 
The \emph{language of first-order coalition logic} ($\CSL$) is defined as 
    \[\varphi := p \mid \neg \varphi \mid \ (\varphi \land \varphi) \mid \assign{t_1,..., t_n} \varphi  \mid \forall x \varphi\]
where $p \in \Ap$, $t_i \in \mathcal{C}\cup\V$, $x \in \V$, and all the usual abbreviations of propositional logic (such as $\vee$, $\to$, $\leftrightarrow$) and conventions for deleting parentheses hold. The existential quantifier $\exists x \varphi$ is defined as 
$\lnot \forall x \lnot \varphi$.
Formula $\assign{t_1,...,t_n} \varphi$ is read as `after the agents execute actions assigned to $t_1\cdots t_n$, $\varphi$ is true', and $\forall x \varphi$ is read as `for all actions $x$, $\varphi$ holds'. Given a formula $\varphi \in \CSL$, the \emph{size of} $\varphi$, denoted by $|\varphi|$, is the number of symbols in $\varphi$. 

\end{definition}



\begin{definition}[Free Variables]
      Given a formula $\varphi$, we define its \emph{set of free variables} $\FV(\varphi)$ by the following cases: 

    \begin{enumerate}
        
        \item If $\varphi\in \Ap$, then $\FV(\varphi)=\emptyset$; 
        \item If $\varphi=\neg \varphi_1$, then $\FV(\varphi)=\FV(\varphi_1)$; 
        \item If $\varphi=\varphi_1 \land \varphi_2$, then $\FV(\varphi)=\FV(\varphi_1)\cup \FV(\varphi_2)$; 
        \item if $\varphi=\assign{t_1,..., t_n }\varphi_1$, then 
       $\FV(\varphi)=\FV(\varphi_1)\cup \set{t_i \mid t_i\in \V} $;
        
     \item if $\varphi= \forall x \varphi_1$, then $\FV(\varphi)=\FV(\varphi_1)\setminus \set{x}$.

    \end{enumerate}

    \noindent A formula $\varphi$ such that $\FV(\varphi)=\emptyset$ is called a \emph{closed formula}, or a \emph{sentence}. 
\end{definition}

\begin{definition}[Kripke Frame]
    A \emph{Kripke frame} is a tuple $\mathcal{F}=\tuple{\Sigma,S,R}$, where $\Sigma$ is a non-empty countable alphabet, $S$ is a non-empty set of states  
    s.t. $\Sigma \cap S = \emptyset$,
    and  $R\subseteq S\times \Sigma \times S$ is a ternary relation, dubbed \emph{transition relation}.  $\mathcal{F}$ is 
    \emph{serial} 
        if for every $s\in S$ and 
        $a\in \Sigma$, there is a  $t\in S$ s.t. $\tuple{s,a,t}\in R$. $\mathcal{F}$ is \emph{functional} 
        whenever for all $s,t,v\in S$ and for every $a\in \Sigma$, if $\tuple{s,a,t}\in R$ and $\tuple{s,a,v}\in R$, then $t=v$. 
    
\end{definition}

\begin{definition}[Concurrent Game Structure]
    A \emph{game frame} is a tuple $\mathcal{G}=\tuple{n,\Ac, \mathcal{D}, S,R }$ with triple $\tuple{\mathcal{D}, S, R}$ being a serial and functional Kripke frame, where:  $n$ is a positive natural number and  $\mathcal{D}$ is a set of tuples of elements of $\Ac$ of length $n$ (elements of this set will be called \emph{decisions}).  
      
    A \emph{Concurrent Game Structure} (CGS) is a pair $\mathfrak{G}=\tuple{\mathcal{G},\mathcal{V}}$, where $\mathcal{G}$ is a game frame, and $\mathcal{V}: \Ap \to \mathcal{P}(S)$ is a \emph{valuation function} assigning to each atomic proposition a subset of $S$. 

    Let $\mathit{Prop} (s) = \{p \in \Ap \mid s \in \mathcal{V}(p)\}$ be the set of all atomic propositions true in state $s$. We define the \emph{size of CGS} $\G$ as $|\G| = n + |\Ac| + |\mathcal{D}| + |S| + |R| + \sum_{s \in S} |\mathit{Prop} (s)|$, where $|\mathcal{D}| = |\Ac|^n$. We call CGS $\G$ \emph{finite}, if $|\G|$ is finite.
\end{definition}



\begin{definition}
    Given a signature $\alpha=\tuple{m,\mathcal{C},\Ap}$,  and  a CGS $\G=\tuple{n,\Ac,\mathcal{D},S,R,\mathcal{V}}$, we say that $\G$ is constructed over $\alpha$ iff  $m=n$ and $\mathcal{C}=\Ac$. 
 \end{definition}

\begin{definition}[Satisfaction]
Let $\varphi$ be a sentence and $\G$ be a CGS that are both constructed over the same signature $\alpha$. 
 The \emph{satisfaction relation} $\G,s\models \varphi$ is inductively defined 
as follows: 
    \begin{alignat*}{3}
        &\G,s\models p &&\text{ iff } &&s\in \mathcal{V}(p)\\
        &\G,s\models \neg \psi &&\text{ iff } &&\G,s\not\models \psi\\
        &\G,s\models \psi \land \chi &&\text{ iff } &&\G,s\models \psi \text{ and } \G,s\models \chi\\
        &\G,s\models \assign{{a_1},..., {a_n}} \psi &&\text{ iff } &&\exists t \in S \text{ s.t. } \tuple{s, {a_1}, ..., {a_n},t}\in R \\ & && &&\text{and }\G,t\models  \psi\\
        &\G,s\models \forall x \psi &&\text{ iff } &&\forall a \in \Ac: \G,s\models \psi[{a}/x]
    \end{alignat*}
\noindent where $a_1,...,a_n$ are constants, and  $\psi[{a}/x]$ denotes 
the result of substituting every occurrence of the variable $x$ with the constant ${a}$ in $\psi$. We will also sometimes write $\vec a$ for $a_1\cdots a_n$.
\end{definition}  


\begin{definition}[Closure of a Formula]
\label{def:clos}
    Given a formula $\varphi$ whose set of free variables is $\set{x_1,\ldots,x_n}$, we denote by $C(\varphi)$ the \emph{closure} of $\varphi$, which is the formula $\forall x_1\cdots \forall x_n \varphi$. 
\end{definition}


\begin{definition}[Validity]\label{def:satopen}

Let $\G$ be a CGS constructed over a signature $\alpha$, and $\varphi$ a formula constructed over $\alpha$. Given a state $s$ of $\G$,
    we write $\G,s\models \varphi$ iff $\G,s\models C(\varphi)$. We say that $\varphi$ is \emph{valid in a CGS} $\G$ (written $\G\models \varphi)$ iff $\G,s\models \varphi$ for every state $s$ of $\G$. Finally, we say that $\varphi$ is $\emph{valid}$ (written $\models \varphi)$ iff it is valid in every CGS constructed over a signature with $n$ agents. Given a set of formulae $X$, we write $\G,s \models X$  if for every  formula $\varphi\in X$, $\G, s\models \varphi$. Finally, we write $X\models \psi$ and we say that $\psi$ is a logical consequence of $X$ iff $\G\models X$ implies $\G\models \psi$ for every CGS $\G$ constructed over the same signature as $\psi$ and formulae  in $X$. 
\end{definition}

\begin{remark}\label{remark:open}
    Note that the truth of open (i.e. not closed) formulae is reduced to the truth of the closed ones via closure (Definition \ref{def:clos}). This approach is fairly standard in first-order logic (see, e.g., \cite{vandalen}). We could also define the truth of a formula w.r.t. an assignment, but this would not affect the results presented here. Our choice simplifies the formal machinery of the paper and makes it more readable. 
\end{remark}

    

The next proposition, that follows straightforwardly from the seriality and functionality of frames, shows that we can give an alternative and equivalent characterisation of the truth of a strategic formula in a state of  a CGS.  

\begin{proposition}
\label{prop:altSem}
    Let $\G = \tuple{n,\Ac, \mathcal{D}, S,R, \mathcal{V} }$ be a CGS, $s \in S$, and $\varphi = \assign{{a_1},..., {a_n}} \psi$, and suppose that both $\G$ and $\varphi$ are constructed over the same signature $\alpha$. 
    Then $\G,s \models \varphi$ iff 
    $\forall t \in S$: $\tuple{s,a_1,\,\ldots, a_n,t}\in R$ implies $\G,t\models \psi$. 
    
     \end{proposition}

   \begin{remark}
      Due to Proposition \ref{prop:altSem}, we 
      we have that
      $\G,s\models \assign{\vec a}\psi$ iff  $\G,t\models \psi$ for the unique $t$ such that $\tuple{s,\vec{a},t}\in R$. 
   \end{remark}
   
Thus, $\CSL$ can be also viewed as an extension of multi-modal logic \cite{thebluebible,hennessy80} for serial and functional frames with first-order quantification over components of arrow labels (i.e. actions).

\begin{example}
As observed in \cite{belardinelli19}, strategy logics are expressive enough to capture 
\textit{Stackelberg equilibrium} (SE). Such an equilibrium is applicable to scenarios where a leader commits to a strategy, and the follower, observing the strategy of the leader, provides her best response. SE is prominent in security games \cite{sinha18}, where the attacker observes the defender committing to a defensive strategy and then decides on the best way to attack (if at all). We can express such a scenario for the case of one-step strategies by the $\CSL$ formula $\forall x_d \exists x_a \forall x_e \assign{x_d, x_a, x_e} \mathit{win_a}$, which intuitively means that for all actions of the defender, the attacker has a counter-action guaranteeing the win for all actions of the environment.    

Similarly to \cite{mogavero10}, with $\CSL$ we can express the existence of deterministic \textit{Nash equilibrium} (NE) for Boolean goals. If $\psi_1, ..., \psi_n$ are goal formulae of agents, we can assert the existence of strategy profile $x_1,...,x_n$ such that if any agent $i$ achieves her goal $\psi_i$ by deviating from $x_1,...,x_n$, then she can also achieve her goal by sticking to the action profile. The existence of such a profile can be expressed by the following $\CSL$ formula:
\begin{align*}
\exists x_1, ..., x_n (\bigwedge_{i=1}^n \exists y_i \assign{x_1,..., y_i, ..., x_n}\psi_i \to \\
\to \assign{x_1,..., x_i, ..., x_n} \psi_i)
\end{align*}

$\CSL$ also allows for strategy sharing. Consider examples of CGSs presented in Figure \ref{fig::exampleCGM}. In structure $\G_1$ we have two states $s$ and $t$, and the agents can transition between the two states if they synchronise on their actions, i.e. execute the same actions. It is easy to verify that $\G_1,s \models \forall x \exists y \assign{x,y} p$ and $\G_1, s \models \forall x \assign{x,x} \lnot p$.
\end{example}

The fact that $\CSL$ is able to capture the Stackelberg and Nash equilibria is significant, since, compared to $\mathsf{SL}$ and its fragments that can also capture \textit{both} equilibria, the complexity of the model checking problem for $\CSL$ is \textit{PSPACE}-complete as shown in the proof of Theorem \ref{thm:model_checking} (compared to the range from \textit{2ExpTime} to non-elementary for various $\mathsf{SL}$'s \cite{mogavero14}). Moreover, the ability to capture the equilibria can have a significant impact on the prospective applications of $\CSL$. In particular, it was argued \cite{vandermeyden19,galimullin22} that $\mathsf{CL}$ is suitable for specification and verification of \textit{atomic swap} 
smart contracts that allow agents to exchange their assets or private information, like passwords, on a blockchain without necessarily trusting each other. We can use $\CSL$ to verify that acting honestly is indeed a NE for a given specification of a contract. Moreover, having the ability to express strategy sharing, we can verify that a swap is still executable in the situation, where a malicious agent that gained access to the communication channel poses as one of the honest ones by executing the same actions\footnote{This is the classic person-in-the-middle attack scenario in cryptography.}.

   \section{Relation to Other Coalition Logics}
\label{sec:expressivity}
In order to appreciate the richness of $\CSL$, we compare the logic to other $\mathsf{CL}$'s. In our comparison we can use two salient features of $\CSL$. First, the logic allows for \textit{arbitrary quantification prefixes for agents' actions}. This includes using the same strategy variable for different agents to capture \textit{strategy sharing}. The second special feature of $\CSL$ is the presence of \textit{explicit action labels} in its syntax.

\begin{definition}[Expressivity]
    Let $\mathsf{L}_1$ and $\mathsf{L}_2$ be two languages, and let $\varphi \in \mathsf{L}_1$ and $\psi \in \mathsf{L}_2$. We call $\varphi$ and $\psi$ \emph{equivalent}, if for any CGS $\G$ and state $s \in \G$: 
    $\G,s \models \varphi$ iff $\G, s \models \psi$.  If for all $\varphi \in \mathsf{L}_1$ there exists an equivalent $\psi \in \mathsf{L}_2$, then $\mathsf{L}_2$ is \emph{at least as expressive as } $\mathsf{L}_1$ ($\mathsf{L}_1 \leqslant \mathsf{L}_2$). And $\mathsf{L}_2$ is \emph{strictly more expressive than } $\mathsf{L}_1$ ($\mathsf{L}_1 < \mathsf{L}_2$) if $\mathsf{L}_1 \leqslant \mathsf{L}_2$ and $\mathsf{L}_2 \not \leqslant \mathsf{L}_1$.
\end{definition}

These two features of $\CSL$, arbitrary quantification prefixes and explicit actions, on their own are not unique in the landscape of logics for strategic reasoning. 
Arbitrary quantification prefixes are a hallmark feature of the whole family of \textit{strategy logics} (see, e.g., \cite{mogavero10,belardinelli19}), to which $\CSL$ belongs. Indeed, $\CSL$ can be considered as a variation of the next-time fragment of $\mathsf{SL}$. 
The idea to refer to actions in the language has also been explored, with 
a prime example being $\mathsf{ATL}$ \textit{with explicit strategies} ($\mathsf{ATLES}$) \cite{walther07}. Another example of such a logic 
is \textit{action logic} 
\cite{borgo07}.

Even though both of the main features of $\CSL$ have been explored in the literature, to our knowledge, $\CSL$ is the first logic for strategic reasoning that combines \textit{both} of them. 



 \paragraph{Coalition logic and quantified coalition logic} The original \textit{coalition logic} ($\mathsf{CL}$) \cite{pauly02}, similarly to $\mathsf{ATL}$, allows only single alternation of quantifiers in coalitional modalities. Moreover, this quantification is implicit. Thus, $\mathsf{CL}$ extends the language of propositional logic with constructs $\langle \! \langle C \rangle \! \rangle \varphi$ that mean `there is a strategy for coalition $C$ to achieve $\varphi$ in the next step'. 
 In \textit{quantified coalition logic} ($\mathsf{QCL}$) \cite{agotnes08}, constructs $\langle \! \langle C \rangle \! \rangle \varphi$ are substituted with $\langle P \rangle \varphi$ meaning `there exists a coalition $C$ satisfying property $P$ such that $C$ can achieve $\varphi$'. Since $\mathsf{QCL}$ is as expressive as $\mathsf{CL}$ (although exponentially more succinct), we will focus only on $\mathsf{CL}$. 
 
 To introduce the semantics of $\mathsf{CL}$, we will denote the choice of actions by coalition $C \subseteq Agt$ with $|Agt| = n$ as $\sigma_C$, and denote $Agt \setminus C$ as $\overline{C}$. Finally, $\sigma_C \cup \sigma_{{\overline{C}}} \in \mathcal{D}$ is a decision. The semantics of $\langle \! \langle C \rangle \! \rangle \varphi$ for a given CGS $\G$ is then defined as 
  \begin{alignat*}{3}
        &\G,s \models \langle \! \langle C \rangle \! \rangle \varphi && \text{ iff } && \exists \sigma_C, \forall \sigma_{\overline{C}} : \G,t \models \varphi \\
        & && &&\text{ with } t \in S \text{ s.t. } \langle s, \sigma_C \cup \sigma_{\overline{C}}, t \rangle \in R.   
\end{alignat*}     
 
The translation from formulas $\mathsf{CL}$ to formulas of $\CSL$ can be done recursively using the following schema for coalitional modalities: 
$tr(\langle \! \langle C \rangle \! \rangle \varphi) \to \exists \vec x \forall \vec y \assign{\vec x, \vec y} tr(\varphi)$, 
where variables $\vec x$ (all different) quantify over actions of $C$, and $\vec y$ (all different) quantify over actions of $Agt \setminus C$. 

At the same time, one cannot refer to particular actions in $\mathsf{CL}$ formulas, as well as express sharing strategies between agents. We can exploit either of these features to show that $\CSL$ is strictly more expressive than $\mathsf{CL}$. 
Indeed, consider a $\CSL$ formula $\exists x \assign{x, x} \lnot p$ meaning that there is an action that \textit{both} agents 1 and 2 should use to reach a $\lnot p$-state. We can construct two CGSs that are indistinguishable by any $\mathsf{CL}$ formulas. At the same time,   $\exists x \assign{x, x} \lnot p$ will hold in one structure and be false in another. 

Consider two structures depicted in Figure \ref{fig::exampleCGM}. 
\begin{figure}[h!]
\centering
\scalebox{0.7}{
\begin{tikzpicture}
\node(-1) at (2,0) {$\G_1$};
\node[circle,draw=black, minimum size=4pt,inner sep=0pt, fill = black, label=below:{$s$}](1) at (0,0) {};
\node[circle,draw=black, minimum size=4pt,inner sep=0pt, , label=below:{$t$}](2) at (4,0) {};

\draw [->,thick](1) to [loop above] node[above, align=left] {$ab, ba$} (1);
\draw [->,thick](1) to [bend right] node[below,align=left] {$aa, bb$} (2);
\draw [->,thick] (2) to [bend right] node[above,align=left] {$aa, bb$} (1);
\draw [->,thick] (2) to [loop above] node[above,align=left] {$ab, ba$} (2);
\end{tikzpicture}
\hspace{6mm}
\begin{tikzpicture}
\node(-1) at (2,0) {$\G_2$};
\node[circle,draw=black, minimum size=4pt,inner sep=0pt, fill = black, label=below:{$s$}](1) at (0,0) {};
\node[circle,draw=black, minimum size=4pt,inner sep=0pt, , label=below:{$t$}](2) at (4,0) {};

\draw [->,thick] (1) to [loop above] node[above, align=left] {$aa, bb$} (1);
\draw [->,thick](1) to [bend right] node[below,align=left] {$ab, ba$} (2);
\draw [->,thick] (2) to [bend right] node[above,align=left] {$ab, ba$} (1);
\draw [->,thick] (2) to [loop above] node[above,align=left] {$aa, bb$} (2);
\end{tikzpicture}
}
\caption{CGSs $\G_1$ and $\G_2$ for two agents and two actions. Propositional variable $p$ is true in black states.}
\label{fig::exampleCGM}
\end{figure}
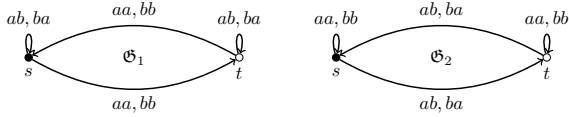 
It is easy to see that $\G_1,s$ and $\G_2,s$ cannot be distinguished by any $\mathsf{CL}$ formula\footnote{These structures are, in fact, in the relation of \textit{alternating bisimulation} \cite{agotnes07}, and hence satisfy the same formulas of $\mathsf{CL}$ and $\mathsf{ATL}$. The discussion of bisimulations for all the logics we mention is, however, beyond the scope of this paper, and we leave it for future work.}. Indeed, both structures agree on the valuation of propositional variable $p$ in corresponding states. Moreover, none of the agents, 1 and 2, can on their own force a transition from state $s$ to state $t$. At the same time, the grand coalition $\{1,2\}$ can match any transition in one structure with a transition with the same effect in the other structure. 
Now, we can verify that $\G_1,s \models \exists x \assign{x, x} \lnot p$ and $\G_2,s \not \models \exists x \assign{x, x} \lnot p$. For the case of  $\G_1,s \models \exists x \assign{x, x} \lnot p$, it is enough to assign action $a$ to $x$ to have $\G_1,s \models \assign{a, a} \lnot p$. To make $\exists x \assign{x, x} \lnot p$ hold in $\G_2, s$, one needs to provide an action that once executed by both agents will force the transition to state $t$. It is easy to see that there is no such an action in $\G_2,s$.

Having the translation from $\mathsf{CL}$ to $\CSL$ on the one hand, and the indistinguishability result on the other, we hence conclude that $\CSL$ is \textit{strictly more expressive} than $\mathsf{CL}$.

\begin{proposition}
\label{prop:cslVScl}
    $\mathsf{CL} < \CSL$.
\end{proposition}

\paragraph{Conditional strategic reasoning and socially friendly $\mathsf{CL}$} With the expressive power of $\CSL$ we can go much further than the classic $\mathsf{CL}$. In particular, we can express in our logic such interesting $\mathsf{CL}$'s like \textit{logic for conditional strategic reasoning} ($\mathsf{ConStR}$) \cite{goranko22}, \textit{socially friendly $\mathsf{CL}$} ($\mathsf{SFCL}$) \cite{goranko18}, and \textit{group protecting $\mathsf{CL}$} ($\mathsf{GPCL}$)  \cite{goranko18}. 

Presenting the semantics of the aforementioned logic is beyond the scope of this paper. However, we would like to point out that all of the logics can be captured by \textit{basic strategy logic} ($\mathsf{BSL}$) \cite{goranko23}, a variant of $\mathsf{SL}$, where each agent has her own associated strategy variable. Differently from $\CSL$, $\BSL$ allows for all standard temporal modalities like `ne\textsf{X}t', `\textsf{U}ntil' and `\textsf{G}lobally'. At the same time, $\BSL$ does not allow for variable sharing and does not explicitly refer to actions or strategies. Moreover, it is conjectured that $\BSL$ does not have a recursive axiomatisation, while $\CSL$ has a finitary complete axiomatisation (see Section \ref{sec:axiom}).

Translations of all coalition logics introduced in this paragraph into formulas of $\mathsf{BSL}$ are presented in \cite{goranko23}, where it is also claimed that $\mathsf{BSL}$ is strictly more expressive than all the aforementioned logics. The translation does not employ any temporal features of $\mathsf{BSL}$ apart from `ne\textsf{X}t', and thus the same translation also works for $\CSL$. Moreover, we can use either strategy sharing or explicit actions to argue that $\CSL$ is strictly more expressive than the considered coalition logics. As an example, an argument for the case of $\mathsf{SFCL}$ is given in the Appendix.


\begin{proposition}
    $\mathsf{ConStR} < \CSL$, $\mathsf{SFCL} < \CSL$, $\mathsf{GPCL} < \CSL$.
\end{proposition}

\paragraph{Action logic} 

A perhaps most relevant to $\CSL$ coalition logic in the literature is \textit{action logic} ($\mathsf{AL}$) \cite{borgo07}, which is a fragment of \textit{multi-agent PDL with quantificaiton} ($\mathsf{mPDLQ}$) \cite{borgo05}. $\mathsf{AL}$ extends the language of propositional logic with so-called \textit{modality markers} $[M]$, which are, essentially, prefixes of size $|Agt| = n$, each element of which can either be a quantifier $Q_i x_i$ with $Q_i \in \{\forall, \exists\}$ or an explicit action \cite{borgo05,borgo05b}. An important feature here is that \textit{there are no repeating variables in modality markers}. Finally, to the best of our knowledge, there is no axiomatisation of  $\mathsf{AL}$.

Given a modality marker $[M]$, we denote by $\sigma_\exists$ a choice by all the existentially quantified agents, by $\sigma_\forall$ a choice by all the universally quantified agents, and by $\sigma_{act}$ explicit actions in the corresponding positions in $[M]$. Then modality markers have the following semantics:
  \begin{alignat*}{3}
        &\G,s \models [M] \varphi && \text{ iff } && \exists \sigma_\exists, \forall \sigma_\forall : \G,t \models \varphi \\
        & && &&\text{ with } t \in S \text{ s.t. } \langle s, \sigma_\exists \cup \sigma_\forall \cup \sigma_{act}, t \rangle \in R.   
\end{alignat*}   
Intuitively, $\G,s \models [M] \varphi$ holds if and only if there is an assignment of actions to all existentially quantified variables in modality marker $M$ such that no matter which actions are assigned to the universally quantified variables, once combined with the explicit actions, the outcome state satisfies $\varphi$. This is in line with the semantics of $\mathsf{CL}$ as we basically choose actions for a coalition (existentially quantified variables) and verify $\psi$ in all possible outcomes given this choice. 

Formulae $[M]\varphi$ of $\mathsf{AL}$ can be translated into formulae of $\CSL$ of the form $\exists \vec x \forall \vec y \assign{t_1, ..., t_n} \varphi$, where $\vec x$ and $\vec y$ with $|\vec x| + |\vec y| \leqslant n$ are possibly empty sequences of variables for the existentially and universally quantified agents respectively, $t_i := x_i$ if there is a quantifier in position $i$ in the modality marker, and $t_i:= a_i$ if there is action $a_i$ in the $i$th posiiton in the modality marker. Also recall that $\mathsf{AL}$ does not allow for sharing strategies (while $\CSL$ does), i.e. all $x_1, ..., x_m$ in the modality marker are unique. 

To show that $\CSL$ is more expressive than $\mathsf{AL}$, we need the following proposition. Its proof utilises the strategy sharing feature of $\CSL$ and can be found in the Appendix. 

\begin{proposition}
\label{lemma:exp}
    $\mathsf{AL}$ is not at least as expressive as $\CSL$.
\end{proposition}

Having the translation from $\mathsf{AL}$ to $\CSL$ on the one hand, and Lemma \ref{lemma:exp} on the other, we can conclude that $\CSL$ is strictly more expressive than $\mathsf{AL}$.

\begin{corollary}
\label{alVScsl}
    $\mathsf{AL} < \CSL$.
\end{corollary}

\paragraph{The expressivity landscape}
In this section we have explored the relationship between $\CSL$ and other notable $\mathsf{CL}$'s from the literature. The overall expressivity landscape of the considered logics is presented in Figure \ref{fig:expressivity}.

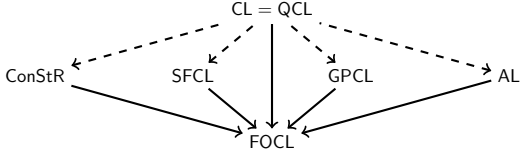
\begin{figure}[h!]
\centering
\begin{tikzpicture}[scale=0.7, transform shape]
\node (constr) at (0,0) {$\mathsf{ConStR}$};
\node (sfcl) at (3,0) {$\mathsf{SFCL}$};
\node (gpcl) at (6,0) {$\mathsf{GPCL}$};
\node (al) at (9,0) {$\mathsf{AL}$};
\node (csl) at (4.5,-1.25) {$\CSL$};
\node (cl) at (4.5,1.25) {$\mathsf{CL} = \mathsf{QCL}$};

\draw[thick, dashed, <-] (constr) to  (cl);
\draw[thick, dashed, <-] (sfcl) to (cl);
\draw[thick,dashed, <-] (gpcl) to (cl);
\draw[thick, dashed, <-] (al) to (cl);
\draw[thick,<-] (csl) to (cl);
\draw[thick,<-] (csl) to (constr);
\draw[thick,<-] (csl) to (sfcl);
\draw[thick,<-] (csl) to (gpcl);
\draw[thick,<-] (csl) to (al);
\end{tikzpicture}
\caption{Overview of the expressivity results. An arrow from $\mathsf{L}_1$ to $\mathsf{L}_2$ means $\mathsf{L}_1 <\mathsf{L}_2$. Dashed arrows represent results from the literature. Solid arrows are new results.}
\label{fig:expressivity}
\end{figure}

\section{Proof Theory}
\label{sec:axiom}

Perhaps the best-known results in the field are complete axiomatisations of $\mathsf{CL}$ \cite{pauly02,goranko13} and $\mathsf{ATL}$ \cite{goranko06} (see \cite{walther06,goranko09} for more constructive approaches). Other completeness results include axiomatisations for logics based on $\mathsf{CL}$ and $\mathsf{ATL}$, like already mentioned $\mathsf{SFCL}$ \cite{goranko18}, $\mathsf{ATLES}$ \cite{walther07}, as well as \textit{epistemic} $\mathsf{CL}$ \cite{agotnes19}, \textit{resource-bounded} $\mathsf{CL}$ \cite{alechina11} and $\mathsf{ATL}$ \cite{nguyen18}, and $\mathsf{ATL}$ \textit{with finitely bounded semantics} \cite{goranko19}, to name a few.

In the context of strategy logics, we have quite an opposite picture. Since the inception of $\mathsf{SL}$ \cite{mogavero10}, its axiomatisation has been an open problem. The same can be said about any of the fragments of $\mathsf{SL}$. 
The lack of axiomatisations of \textit{any} (fragment of) $\mathsf{SL}$ can be traced back to the two main features of the logic: quantification over strategies and arbitrary quantification prefixes. 



Indeed, arbitrary alternation of quantifiers in $\mathsf{SL}$ is quite different from the fixed quantification prefix of $\mathsf{CL}$ and $\mathsf{ATL}$ that allow only prefixes $\exists \forall$ and $\forall \exists$. Secondly, quantification over strategies\footnote{A \emph{(memoryless) strategy} for an agent $i \in n$ is a function $\sigma_i: S \to \Ac$.}  in $\mathsf{SL}$ is essentially a second-order quantification over functions. We believe that these two features combined are the root cause of the fact that no complete axiomatisations of (fragments of) $\mathsf{SL}$ 
have been proposed so far. 

In $\CSL$ we focus on arbitrary quantification prefixes. To solve this sub-problem, we consider only ne$\mathsf{X}$t-time modalities $\assign{t_1\cdots t_n} \varphi$ and deal with the immediate outcomes of agents' choices. 
This allows us, in particular, to consider quantification over actions rather than strategies. 
Hence, quantification in $\CSL$ is a first-order quantification, instead of the second-order quantification of $\mathsf{SL}$. 

In our proof, we take as inspiration the completeness proof for \textit{first-order modal logic} ($\mathsf{FOML}$) with constant domains \cite{Garson1984}. Our construction is quite different, though, as in $\mathsf{FOML}$ variables appear in $n$-ary predicates, and in $\CSL$ variables are placeholders for transition labels.  

\subsection{Axiomatisation of $\CSL$}

\begin{definition}[Axiomatisation]
    The \emph{axiom system} for $\CSL$ consists of the following axiom schemata and rules, where $\vec{t}=t_1,\ldots,t_n$ for $n\geqslant 1$, and $t$ and each $t_i$ are either a variable or a constant. 
    $$
\begin{array}{c
l}

    \mathsf{PC}  & \text{Every propositional tautology}  
  \\
     \mathsf{K} & ( \assign{\vec{t}} \varphi \land \assign{\vec{t}}\psi) \leftrightarrow \assign{\vec{t}}(\varphi \land \psi )
   \\
   \mathsf{\mathsf{N} } & \neg \assign{\Vec{t}}\varphi \leftrightarrow \assign{\vec{t }}\neg \varphi 

   \\
   \mathsf{E} & \forall x \varphi \to \varphi[t/x]
\\
\mathsf{B} & \forall x \assign{\vec t} \varphi \imp \assign{\vec t} \forall x \varphi, \text{s.t. } t_i \neq x \text{ for all } t_i 
\\
\mathsf{MP} & \text{From } \varphi, \varphi \to \psi, \text{ infer } \psi
\\
\mathsf{Nec} & \text{From } \varphi, \text{ infer } \assign{\vec t }\varphi
\\
\mathsf{Gen} & \text{From } \varphi\to \psi[t/x] , \text{ infer } \varphi\to \forall x \psi, \text{ if $t \not \in \varphi$} 
\end{array}
$$
\noindent An axiomatic derivation $\pi$ is a finite sequence of formulae $\varphi_1,\ldots, \varphi_m$ where for each $i\leqslant m$:  either $\varphi_i$ is an instance of one of the axiom schemata of $\CSL$, or it is obtained from some preceding formulae in the sequence using rules $\mathsf{MP}$, $\mathsf{Nec}$, or $\mathsf{Gen}$.
We write   $\vdash\varphi$ and say that  $\varphi$ is \emph{$\CSL$ derivable} (or simply derivable)  iff there is a derivation $\pi$ whose last element is $\varphi$. Given a set of formulae $X$, we write $X \vdash \varphi$ iff there is a finite subset $Y $ of $X$ such that 
$\vdash \bigwedge Y \imp \varphi$.
\end{definition}

We will freely use the following proposition in the rest of the paper. Its proof is standard, and we omit it for brevity. 

\begin{proposition}\label{prop.first}
    The following formulae are $\CSL$ derivable, where $\vec{t}=t_1,\ldots,t_n$, and each $t_i$ is either a constant or a variable: 
    \begin{enumerate}
        \item $\assign{\vec{t}}(\varphi\imp \psi)\imp (\assign{\vec{t}}\varphi) \imp (\assign{\vec{t}}\psi); $
        \item $\forall x (\varphi \imp \psi)\imp (\varphi \imp \forall x \psi)$ with $x\notin \FV(\varphi);$
        \item $\exists z (\varphi \imp \forall y \varphi )$ with $z\notin \FV(\forall y \varphi).$
      \end{enumerate}
      Moreover, if $\varphi\imp \psi$ is derivable, so is $\assign{\vec{t}}\varphi \imp \assign{\vec{t}} \psi $.
\end{proposition}

\begin{lemma}\label{lemma:sound}
    Each axiom schema of $\CSL$ is valid and each rule of $\CSL$ preserves validity.
\end{lemma}

The proof of Lemma \ref{lemma:sound} is done by the application of the definition of the semantics, and it can found in the Appendix.

Our completeness proof is based on the canonical model construction, where states are maximal consistent sets with the $\forall$-property.

\begin{definition}[Maximal Consistent Sets]
Let $Z$ be a set of $\CSL$ sentences over a given signature and $X\subseteq Z$. We say that: (i)$X$ is \emph{consistent} iff $X\not\vdash \bot$, (ii) $X$ is \emph{maximally consistent} (MCS) iff it is consistent and there is no other consistent set of sentences $Y\subseteq Z$ s.t. $X\subset Y$, 
and (iii) $X$ has the \emph{$\forall$-property} iff for every formula $\varphi$ over the same signature as $Y$  and variable $x$, there is a constant $a$ such that $\varphi[a/x]\to \forall x \varphi\in X$, where $\varphi[a/x]$ is closed. We will call a set satisfying all the three requirements $\forall$-MCS.

\end{definition}

Let $\alpha=\tuple{n,\mathcal{C},\Ap}$ be a signature. We denote by $\alpha^\star$ the signature $\tuple{n,\mathcal{C}\cup \mathcal{C}^\star,\Ap}$ where $\mathcal{C}^\star$ is countably infinite, and $\mathcal{C}\cap \mathcal{C}^\star = \emptyset$. 

Next lemma shows  that each consistent set of sentences over a given signature $\alpha$  can be extended to a consistent set of sentences over $\alpha^\star$ having the $\forall$-property.  Its proof follows the standard technique in $\mathsf{FOML}$ \cite{Cresswell1996-CREANI-3}, and can be found in the Appendix. 

\begin{lemma}\label{lemma:expConst}
    If $X$ is a consistent set of sentences over a given signature $\alpha$, then there is a consistent set of sentences $Y$ over $\alpha^\star$ such that $X\subseteq Y$, and $Y$ has the $\forall$-property.  
\end{lemma}

The proof of the following lemma (Lindenbaum Lemma) is standard, and we omit it for brevity.

\begin{lemma}
\label{lemma:mcs}
Let $X$ be a consistent set of sentences 
over a given signature, then there is 
an MCS $Y$ over the same signature such that $X\subseteq Y$.
\end{lemma}

The next two lemmas will be instrumental in the proof of the Truth Lemma, and showing that the canonical model we are to define in this proof is indeed a CGS.

\begin{lemma}\label{lemma:diamond}
Let $X$ be a consistent set of sentences over a given signature and let $\vec{a}$ be a tuple of constants,   
 then the set $Y_{\vec{a}}=\set{\psi\mid \assign{\vec{a}}\psi\in X}$ is also consistent.
\end{lemma}
\begin{proof}
If $Y_{\vec{a}}$ is empty, the result is trivially valid. Assume that $\varphi \in Y_{\vec{a}}$ and  suppose
    towards a contradiction, that set $Y_{\vec{a}}$ is not consistent. This implies that $(\psi_1\land \cdots \land \psi_m)\to \neg \varphi$ for some finitelty many $\psi_1,\ldots, \psi_m$ in $Y_{\vec{a}}$. Using Proposition \ref{prop.first}(1) and propositional reasoning, we can then derive $\assign{\vec{a}}\psi_1 \land \cdots \land \assign{\vec a}\psi_m\imp \assign{\vec a} \neg \varphi  $. Since $\assign{\vec{a}}\psi_i\in X$, we conclude by $\mathsf{MP}$ that $X \vdash \assign{\vec a} \neg \varphi$. Then, by $\mathsf{N}$ and $\mathsf{MP}$  we can further derive $X\vdash \neg \assign{\vec{a}}\varphi$, which contradicts $\assign{\vec a} \varphi\in X$. 
\end{proof}


\begin{lemma}
\label{lemma:lindy}
Let $X$ be a $\forall$-MCS over a given signature containing infinitely many constants. Then there exists a $\forall$-MCS $Y$ over the same signature such that $Z=\set{\psi \mid \assign{\vec{a}} \psi \in X} \subseteq Y$.
\end{lemma}
\begin{proof}
If $Z = \emptyset$, the lemma trivially holds. 
So assume $\varphi \in Z$, and let $E$ be an enumeration of all sentences of the form $\forall x \xi$, and $C$ an enumeration of the constants in the given signature. We define a sequence of sentences $\theta_0, \theta_1, \ldots$ where $\theta_0 = \varphi$, and given $\theta_n$, we set $\theta_{n+1} = \theta_n \land (\xi[a/x] \imp \forall x \xi)$, where $\forall x \xi$ is the $(n+1)$-th formula in $E$, and $a$ is the first constant in $C$ such that $$(\star)\quad Z \cup \set{\theta_n \land (\xi[a/x] \imp \forall x \xi)} \text{ is consistent.}$$

Let $Y = Z \cup \set{\theta_n \mid n \in \mathbb{N}}$. Clearly $Y$ has the $\forall$-property, and it is consistent if $Z \cup \set{\theta_n}$ is consistent for every $n \in \mathbb{N}$. To show this, we prove that if $Z \cup \set{\theta_n}$ is consistent, then there always exists a constant $a$ satisfying $(\star)$. The set $Z \cup \set{\theta_0} = Z$ is consistent by Lemma \ref{lemma:diamond}. 

Suppose, towards a contradiction, that $Z \cup \set{\theta_n}$ is consistent but for every constant $a$, the set $Z \cup \set{\theta_n \land (\xi[a/x] \imp \forall x \xi)}$ is inconsistent. Then for each constant $a$, there exist finitely many formulas $\psi_1^a, \ldots, \psi_m^a$ in $Z$ such that $(\psi_1^a \land \cdots \land \psi_m^a) \imp (\theta_n \imp \neg(\xi[a/x] \imp \forall x \xi))$ is derivable. From this, by the rules of $\CSL$, it follows that $(\assign{\vec{a}} \psi_1^a \land \cdots \land \assign{\vec{a}} \psi_m^a) \imp \assign{\vec{a}}(\theta_n \imp \neg(\xi[a/x] \imp \forall x \xi))$ is derivable. Since $\psi_i^a \in Z$ implies $\assign{\vec{a}} \psi_i^a \in X$, we conclude that (i) $\assign{\vec{a}}(\theta_n \imp \neg(\xi[a/x] \imp \forall x \xi)) \in X$ for every constant $a$.

Let $z$ be a variable that occurs neither in $\theta_n$ nor in $\xi$. Consider the sentence $\forall z\, \assign{\vec{a}}(\theta_n \imp \neg(\xi[z/x] \imp \forall x \xi))$. From the $\forall$-property of $X$ and (i), it follows that $\forall z\, \assign{\vec{a}}(\theta_n \imp \neg(\xi[z/x] \imp \forall x \xi)) \in X$. By axiom $\mathsf{B}$, this implies $\assign{\vec{a}} \forall z (\theta_n \imp \neg(\xi[z/x] \imp \forall x \xi)) \in X$, and thus, by (2) of Prop. \ref{prop.first}, we have (ii) $\assign{\vec{a}}(\theta_n \imp \forall z \neg(\xi[z/x] \imp \forall x \xi)) \in X$. Since $\exists z (\xi[z/x] \imp \forall x \xi)$ is derivable in $\CSL$, applying rule $\mathsf{Nec}$ gives $\assign{\vec{a}} \exists z (\xi[z/x] \imp \forall x \xi) \in X$. From this, together with (ii), and using Proposition \ref{prop.first}, 
we conclude $\assign{\vec{a}} \neg \theta_n \in X$. By the construction of $Z$, this implies $\neg \theta_n \in Z$, which contradicts the assumption that $Z \cup \set{\theta_n}$ is consistent.
\end{proof}






    

\begin{definition}[Canonical Model]
\label{def:can_model}
Given a signature $\alpha=\tuple{n,\mathcal{C}, \Ap}$, 
the canonical model over $\alpha$ is the tuple $\G^C=\tuple{n, \Ac^C,\mathcal{D}^C, S^C, R^C,$    $\mathcal{V}^C}$, where: 

\begin{itemize}
    \item $\Ac^C=\mathcal{C}\cup \mathcal{C}^\star$;
   
    \item $\mathcal{D}^C = {Ac^C}^n $;
 \item $S^C = \{X \mid X \text{ is a $\forall$-MCS over }$ $\alpha^\star\}$;
    \item for every $\vec{a}\in \mathcal{D}^C$, $\tuple{X,\vec{a}, Y}\in R^C$ iff for every sentence $\varphi$  we have that $\varphi\in Y$ implies $\assign{\vec{a}}\varphi \in X$;
    \item $X \in \mathcal{V}^C(p)$ iff $p\in X$ for all $p \in \Ap$. 
\end{itemize}
\end{definition}

The proof of the next proposition is in the Appendix.

\begin{proposition}\label{prop:existforall}
    For all states $X, Y \in S^C$ 
    and for every decision $\vec{a} \in \mathcal{D}^C$, it holds that $\tuple{X,\vec{a},Y}\in R^C$ iff for every sentence $\varphi$, $\assign{\vec{a}}\varphi\in X $ implies $\varphi\in Y$
\end{proposition}
    

Now we are ready to show that $\G^C$ is indeed a CGS (proof in the Appendix), and then prove the Truth Lemma. 

\begin{proposition}
The canonical model $\G^C$ is a CGS.
\end{proposition}

\begin{lemma}[Truth Lemma] 
\label{lemma:truth}
For any state $X \in S^C$ 
and for any sentence $\varphi$, we have that $\G^C, X \models \varphi$ iff $\varphi \in X$.  
\end{lemma}
\begin{proof}

    The proof is by induction on $\varphi$. The \textit{base case} $\varphi = p$ 
   follows from the definition of $\mathcal{V}^C$. 
    Boolean cases follow from the induction hypothesis (IH) and the properties of MCSs.

\textit{Case} $\varphi = \assign{\vec{a}} \psi$. 
Let $\G^C,X\models \varphi$. By the definition of semantics, this means that there is a $Y$ such that $\tuple{X,\vec{a},Y}\in R^C$ and $\G^C,Y\models \psi$. The latter is equivalent to $\psi\in Y$ by the IH, and by the definition of $R^C$ we conclude that $\varphi\in X$.

Let $\varphi \in X$. By Lemma \ref{lemma:lindy}, there is a maximal consistent set of sentences $Y$ over $\alpha^\star$ that has the $\forall$-property and such that $\set{\psi}\cup\set{\theta\mid \assign{\vec{a}}\theta\in X}\subseteq Y$. By Proposition \ref{prop:existforall} this means $\tuple{X,\vec{a},Y}\in R^C$, which, in conjunction with the fact that $\psi \in Y$, is equivalent to $\G^C,X \models \varphi$ by the IH.

\textit{Case} $\varphi = \forall x \psi$. If $\G^C,X\models \varphi$, then, by the IH, it holds that (i) $\psi[a/x]\in X$ for every $a\in \Ac^C$ . Now, assume towards a contradiction that $\varphi\not\in X$.  Since $X$ is maximally consistent, we have that $\neg\forall x \psi\in X$. Moreover, since $X$ has the $\forall$-property,  there is a constant $a$ such that $\psi[a/x]\to \forall x \psi \in X$. Then by (i) it follows that 
$\forall x \psi\in X$, which contradicts $\forall x \psi \not \in X$.

Suppose that $\varphi\in X$, which implies, by axiom $\mathsf{E}$ and $\mathsf{MP}$, that $\psi[a/x] \in X$ for every $a\in \Ac^C$. By the IH, we conclude that $\G^C,X\models \varphi [a/x]$ for every $a\in \Ac^C$, which is equivalent to $\G^C,X\models \forall x \varphi$ by the definition of semantics. 
\end{proof}

We finally prove the 
completeness of $\CSL$.
\begin{theorem}
    For every set of formulae $X$ and every formula $\varphi$, we have that $X\vdash \varphi$ iff $X\models \varphi$.
\end{theorem}
\begin{proof}
    Let $X\not\vdash \varphi$. This means that $X\cup\set{\neg\varphi}$ is consistent, and, by Lemmas \ref{lemma:expConst} and \ref{lemma:mcs}, there is a 
    $\forall$-MCS $Z$, 
    such that $X\cup\set{\neg\varphi} \subseteq Z$. As $\neg\varphi \in Z$, it holds that $\varphi\notin Z$, and by the truth lemma we have that $\G^C,Z\models X$ and $\G^C,Z\not\models \varphi$.  
\end{proof}

    

\section{Complexity Profile of $\CSL$}

\label{sec:mc}
Now we turn to the complexity profile of $\CSL$, and show that 
the complexity of the model checking problem \textit{PSPACE}-complete and that the satisfiability problem is undecidable.

\paragraph*{Model Checking}
    Let $\G = \tuple{n,\Ac, \mathcal{D}, S,R, \mathcal{V} }$ be a finite CGS, $s \in S$, and closed formula $\varphi \in \CSL$ constructed over a signature of $\G$. The \emph{local model checking problem} for $\CSL$ consists in computing whether $\G, s \models \varphi$.

\begin{theorem}
\label{thm:model_checking}
    The model checking problem for $\CSL$ is PSPACE-complete.
\end{theorem}

\begin{proof}
To show that the model checking problem for $\CSL$ is in \textit{PSPACE}, we provide an alternating recursive Algorithm \ref{cslMC}\footnote{For brevity, we omit Boolean cases and the whole algorithm is available in the Appendix} that takes as an input a finite CGS $\G$, state of the CGS $s$, and a closed formula $\varphi$. The formula $\varphi$ is provided in negation normal form (NNF), i.e. in equivalent rewriting, where all negations are pushed inside and appear only in front of propositional variables. To convert $\varphi$ into the equivalent NNF formula, we can use propositional equivalences, interdefinability of quantifiers, and the validity  $\neg \assign{\Vec{t}}\varphi \leftrightarrow \assign{\vec{t }}\neg \varphi$. The size of a formula in NNF is at most linear in the size of the original formula.
\begin{algorithm}
	\caption{An algorithm for model checking $\CSL$} \label{cslMC} 
	\begin{algorithmic}[1] 		
		\Procedure{MC}{$\G, s, \varphi$}		
       \Case {$\varphi = \assign{a_1,...,a_n} \psi$}
       \State{\textbf{guess} $t \in S$ such that   $\tuple{s, a_1, ..., a_n, t} \in R$}
       \State{\textbf{return} $\textsc{MC} (\G, t, \psi)$}
        \EndCase
         \Case{$\varphi = \exists x \psi$}
        \State{\textbf{guess} $a \in \Ac$  }
        \State{\textbf{return} $\textsc{MC} (\G,s,\psi[a/x])$}
        \EndCase
        \Case{$\varphi = \forall x \psi$}
        \State{\textbf{universally choose} $a \in \Ac$  }
        \State{\textbf{return} $\textsc{MC} (\G,s,\psi[a/x])$}
        \EndCase
   \EndProcedure
	\end{algorithmic}
\end{algorithm}
The correctness of the algorithm follows from the definition of the semantics. Its termination follows from the fact that every recursive call is run on a subformula of smaller size. Moreover, each call of the algorithm takes at most polynomial time, and hence it is in \textit{APTIME}. From the fact that \textit{APTIME} = \textit{PSPACE} \cite{alternation}, we conclude that the model checking problem for $\CSL$ is in \textit{PSPACE}.

    The hardness can be shown by the reduction from the satisfiability of quantified Boolean formulas (see Appendix). 
\end{proof}

\begin{remark}
    The model checking problem for a related \textit{$\mathsf{SL}$ with simple goals} ($\mathsf{SL[SG]}$) is $P$-complete \cite{belardinelli19}. 
This is due to the fact that in $\mathsf{SL[SG]}$ the quantification prefix and the operators for assigning strategies to agents always go together. Hence, for example, the $\CSL$ formula over two agents $\theta:=\forall x \exists y \forall z (\assign{z,x} \varphi \land \assign{z,y} \psi \land \assign{x,z}\chi)$ cannot be expressed in $\mathsf{SL[SG]}$. The higher complexity of $\CSL$ stems from the fact that quantifiers and strategy assignments are less rigid than in $\mathsf{SL[SG]}$, and thus $\CSL$ is closer to the full $\mathsf{SL}$ in this regard. 
\end{remark}

\paragraph*{Satisfiability}
    Let $\varphi \in \CSL$ be a closed formula. The \emph{satisfiability problem} for $\CSL$ consists in determining whether there is a CGS $\G,s$ such that $\G,s \models \varphi$.


\begin{theorem}
    The satisfiability problem for $\CSL$ is undecidable.
\end{theorem}

The undecidability can be shown by employing the construction for $\mathsf{SL}$ from \cite{mogavero10,mogavero16} using the reduction from the classic tiling problem \cite{wangTile}. In their construction, the authors use only formulae of  the next-time fragment of $\mathsf{SL}$, and the proof can be adapted for $\CSL$ (see the Appendix for details).

\begin{remark}
    A knowledgeable reader may point out that the proof in \cite{mogavero10,mogavero16} employed the reduction from a more complex \emph{recurring tiling problem} \cite{harel83}. The problem is known to be $\Sigma^1_1$-complete, and this in particular implies that the next-time fragment of $\mathsf{SL}$, and hence $\CSL$, is not recursively axiomatisable. This is at odds with the axiomatisation of $\CSL$ presented in this paper. However, after a closer inspection, it turned out that the  $\Sigma^1_1$-hardness proof provided in \cite{mogavero10,mogavero16} is incomplete (though the standard, non-recurring, tiling construction stands, and hence the (standard) undecidability). As of now, no fix to this problem has been presented, and hence the existence of a recursive axiomatisation of $\mathsf{SL}$ is now an open question. Note, again, that the (standard) undecidability still holds\footnote{The gap in the proof of non-axiomatisability of $\mathsf{SL}$ was acknowledged and corroborated by the authors of \cite{mogavero10,mogavero16} in personal communication. The general sentiment is that $\mathsf{SL}$ is still not recursively axiomatisable, but to show this result, one will have to employ richer features of $\mathsf{SL}$, beyond its next-time fragment.}. 
\end{remark}

\section{Discussion}
\label{sec:conclusion}

We introduced \textit{first-order coalition logic} ($\CSL$), which combines features of both $\mathsf{CL}$ and $\mathsf{SL}$, and, additionally, allows for explicit action labels in the syntax. With $\CSL$ we have solved several exciting problems. First, we showed that it is strictly more expressive than other known $\mathsf{CL}$'s, and that its model checking problem is \textit{PSPACE}-complete. We then also argued that the satisfiability problem for $\CSL$ is undecidable, pointing out an incomplete result in the foundational $\mathsf{SL}$ paper \cite{mogavero10} and thus reopening the question of whether $\mathsf{SL}$ is recursively axiomatisable.
Moreover, we provided a sound and complete axiomatisation of $\CSL$. This is significant, since, to the best of our knowledge, it is \textit{the first axiomatisation of any strategy logic}. 



There is a plethora of open research questions that one can tackle building on our work. Perhaps the most immediate one is finding an axiomatisation of an extension of $\CSL$ with $\mathsf{LTL}$ modalities. In such a way, we would be able to advance towards axiomatisations of such rich fragments of $\mathsf{SL}$ as \textit{one-goal} $\mathsf{SL}$ \cite{mogavero16} and \textit{flat conjunctive-goal} $\mathsf{SL}$ \cite{acar19}. It is also quite interesting to consider $\CSL$ in the context of imperfect information (see \cite{agotnes15} for an overview). 

While dealing with the undedicability of $\CSL$, we mentioned the next-time fragment of $\mathsf{SL}$. To the best of our knowledge, such a fragment has never been singled out and studied before. Hence, it is tempting to look into the variations of this fragment\footnote{There is a nuance in how we can define the next-time fragment of $\mathsf{SL}$. Two obvious candidates are the fragment, where every ne$\mathsf{X}$t modality is immediately preceded by an assignment (similar to $\mathsf{ATL}$), and the fragment, where we do not have such a condition (similar to $\mathsf{ATL}^\star$).}, identify the axiomatisable ones, and have a proper comparison of the latter with $\CSL$. 

As $\mathsf{STIT}$ logics \cite{horty01} admit CGS semantics \cite{boudou18,broersen15}, another avenue of exciting further research is establishing the exact relation between $\CSL$ and variants of $\mathsf{STIT}$ logics like \textit{group} $\mathsf{STIT}$ \cite{herzig08,lorini11}.

\section*{Acknowledgements}
The authors would like to thank Valentin Goranko for the discussion of the preliminary ideas of this paper, Fabio Mogavero for the discussion of the undecidability of $\mathsf{SL}$, and the anonymous reviewers of IJCAI for their superb comments and suggestions.

\bibliographystyle{named}
\bibliography{cslref}

\clearpage
\appendix
\section*{Technical Appendix}

\subsection*{Relation to Other Formalisms}
\paragraph*{$\mathsf{SFCL} < \CSL$}

We provide an argument for $\mathsf{SFCL}$ that extends the language of propositional logic with constructs $\langle \! \langle  C \rangle \! \rangle (\varphi;\psi_1,...,\psi_k)$ meaning that `coalition $C$ can achieve $\varphi$ while also enabling $\overline{C}$ to achieve any of $\psi_1,...,\psi_k$ (via a suitable joint action)'. 

Formally, the semantics is defined as 
\begin{gather*}
    \G,s \models \langle \! \langle C \rangle \! \rangle (\varphi; \psi_1, ..., \psi_k) \text{ iff }\\
    \exists \sigma_C (\forall \sigma_{\overline{C}} : \G,t \models \varphi \text{ and } \forall \psi_i, \exists  \sigma_{\overline{C}}:\G,u \models \psi) \text{ with } \\ t,u \in S \text{ s.t. } \langle s, \sigma_C \cup \sigma_{\overline{C}}, w \rangle \in R
\text{ and } w \in \{t,u\}.
\end{gather*}

Now, let us have another look at CGSs in Figure \ref{fig::exampleCGM}. Recall that these structures are distinguished by the $\CSL$ formula $\exists x \assign{x,x} \lnot p$. We claim that no formula of $\mathsf{SFCL}$ can distinguish $\G_1,s$ from $\G_2,s$. An informal sketch of the induction-based argument is as follows. For purely propositional formulas it is clear that $\G_1,w$ and $\G_2,w$ with $w\in \{s,t\}$ satisfy the same formulas. Now, let us consider socially-friendly coalitional modalities $\langle \! \langle  C \rangle \! \rangle (\varphi;\psi_1,...,\psi_k)$ For the case of grand coalition $C = \{1,2\}$, it is easy to verify that any move in $\G_1,w$ can be matched by a corresponding move $\G_2,w$ to satisfy $\varphi$. Clearly, these transitions will require different actions by agents, but since we do not have access to action labels in $\mathsf{SFCL}$, we are not able to spot the difference. 

For the case of single agents, observe that yet again, every choice of, let's say, agent 1 in one structure can be matched by a choice in the other structure to the same effect. Indeed, whatever agent 1 chooses in $\G_1, s$, $a$ or $b$, she can only satisfy some $\varphi$ that holds in both states $s$ and $t$ (due to the fact that the outcome is determined by what agent 2 chooses as well). Similarly in $\G_2, s$. Now, goals $\psi_i$ of agent 2 can either be satisfied in state $s$, state $t$, or both states. Hence, by the construction of CGSs, in both $\G_1$ and $\G_2$ for each choice of agent 1, agent 2 has an action to either stay in the current state or force the transition to another state. That the outcome of the corresponding transitions satisfy $\psi_i$ follows from the induction hypothesis. 
\newline
\newline
\textbf{Proposition 1.}  $\mathsf{AL}$ is not at least as expressive as $\CSL$.

\begin{proof}
  Consider $\exists x \assign{x, x} \lnot p \in \CSL$, and assume towards a contradiction that there is an equivalent $\varphi \in \mathsf{AL}$. Since we have a countably infinite set of constants $\mathcal{C}$ (and hence actions) at our disposal  and due to the fact that $\varphi$ is finite, we can assume that there are actions $a$ and $b$ that do not appear explicitly in $\varphi$. 

Now, consider two concurrent game structures defined over two agents and two actions in Figure \ref{fig::exampleCGM}. As we have already seen in our argument for Proposition \ref{prop:cslVScl}, $\G_1,s \models \exists x \assign{x, x} \lnot p$ and $\G_2,s \not \models \exists x \assign{x, x} \lnot p$. What is left to show is that $\varphi$ cannot distinguish the two structures, i.e. $\G_1,s \models \varphi$ if and only if $\G_2, s \models \varphi$. 

       The proof is by induction on the complexity of $\varphi$. As the \textit{Base Case}, by the construction of the structures we have that $\G_1,w \models p$ if and only if $\G_2,w \models p$ for $w \in \{s,t\}$ and all $p \in \Ap$. 

    \textit{Induction Hypothesis.} $\G_1,w \models \psi$ if and only if $\G_2, w \models \psi$  for $w \in \{s,t\}$ and for all strict subformulas $\psi$ of $\varphi$.
    
    Boolean cases follow straightforwardly by the induction hypothesis. What is left is the case of modality markers. 

   \textit{Case } $\varphi := [M]\psi$. First, recall that we assume that actions $a$ and $b$ do not appear explicitly in $\varphi$. It is enough to verify four forms of modality markers corresponding to all possible combinations of quantifiers over actions for two agents. Let $\varphi = [\exists x, \exists y] \psi$. It is easy to see that  $\G_1,w \models [\exists x, \exists y] \psi$ if and only if $\G_2,w \models [\exists x, \exists y] \psi$ as in both CGSs the grand coalition of agents $\{1,2\}$ has the full control over which transitions to force. Hence, each move in $\G_1,w$ to a $\psi$-state $v \in \{s,t\}$ can be matched by a move in  $\G_2,w$ to the same $\psi$-state $v$, where we will have $\G_1,v \models \psi$ if and only if $\G_2,v \models \psi$ by the induction hypothesis. 
   The remaining cases for modality markers can be shown similarly.
\end{proof}

\subsection*{Soundness and Completeness}

\textbf{Lemma 1.}
    Each axiom schema of $\CSL$ is valid and each rule of $\CSL$ preserves validity.

\begin{proof}
 For the sake of simplicity, we only consider closed instances of the axiom schemata. Validity of other axiom schemata and the soundness of the rules of inference can be shown similarly.
 
 ($\mathsf{N}$). Suppose that $\G,s\models \neg \assign{\vec{{a}}}\varphi$. By the definition of the semantics, this means that $\G,s\not \models \assign{\vec{{a}}} \varphi$, i.e. for each   $t\in S$ if   $\tuple{s,\vec{a},t}\in R$,  then we have that $\G,t\not\models \varphi$. From the seriality and functionality of $R$, we can conclude that there is exactly one such $t$, and thus $\G,s\models \assign{\vec{{a}}}\neg \varphi $. For the converse direction, suppose that $\G,s\models \assign{\vec{{a}}}\neg \varphi$. This means that there is a $t$ such that $\tuple{s,\vec{a},t}\in R$ and $\G,t\not\models \varphi$. By functionality of $R$ there is no other $t$ related to $s$ by means of $\vec{a}$, and thus we can conclude that $\G,s\models \neg \assign{\vec{{a}}}\varphi$. 

 ($\mathsf{B}$). Assume that $\G,s \models \forall x \assign{\vec t} \varphi$, where $x$ is different from every $t_i$. Since the formula is closed, this is just $\G,s \models \forall x \assign{\vec{ {a}}} \varphi$ for some $\vec{a}\in \mathcal D$. By the truth definition, this is equivalent to $\G,s \models \assign{\vec{{a}}} (\varphi [{b}/x])$ for every $b\in \Ac$, which means $\G,s\models \assign{\vec{ a}} \forall x \varphi $.
\end{proof}

\textbf{Lemma 2.}
    If $X$ is a consistent set of sentences over a given signature $\alpha$, then there is a consistent set of sentences $Y$ over $\alpha^\star$ such that $X\subseteq Y$, and $Y$ has the $\forall$-property.

\begin{proof}
    Let $E$ be an enumeration of sentences of the form $\forall x \varphi$ over $\alpha^\star$. We define a sequence of sets of sentences $Y_0,Y_1,\ldots$ with $Y_0=X$ and $Y_{n+1}=Y_n \cup \set{ \varphi[a/x]\to \forall x \varphi}$ where $\forall x \varphi$ is the $n+1$-th sentence in $E$, and $a$ is the first constant in the enumeration occurring neither in $Y_n$ nor in $\varphi$. Since $Y_0$ is over $\alpha$, $Y_n$ is obtained by the addition of $n$ sentences over $\alpha^\star$, and $\alpha^\star$ includes a countably infinite set of new constants, we can always find such an $a$. 
    
    Now we show that $Y_{n+1}$ constructed in the described way is consistent. For this, assume towards a contradiction that $Y_n$ is consistent and $Y_{n+1}$ is not. This means that there is a finite set of sentences $U\subseteq Y_n$ such that $U\cup \set{\varphi [a/x]\to \forall x \varphi}\vdash \bot$. By the propositional reasoning we thus obtain that (i) $U\vdash \varphi [a/x]$ and (ii) $U\vdash \neg \forall x \varphi$. Since $a$ does not appear in $Y_n$, we can use the  $\mathsf{Gen}$ rule of inference and conclude that $U\vdash \forall x \varphi$. In conjunction with (ii) this amounts to the fact that $Y_n$ is not consistent, and hence we arrive at a contradiction.
    
    Define $Y$ as $\bigcup_{n\in \mathbb{N} }Y_n$. It is now easy to see that $Y$ is consistent and has the $\forall$-property. 
\end{proof}

\textbf{Proposition 6.}
    For all states $X, Y \in S^C$ 
    and for every decision $\vec{a} \in \mathcal{D}^C$, it holds that $\tuple{X,\vec{a},Y}\in R^C$ iff for every sentence $\varphi$, $\assign{\vec{a}}\varphi\in X $ implies $\varphi\in Y$

\begin{proof}
    Left-to-right: suppose $\tuple{X,\vec{a},Y}\in R^C$  and $\varphi\not\in Y$. We need to show that $\assign{\vec{a}}\varphi\notin X$. Since $Y$ is maximally consistent, we have that $\neg\varphi\in Y$. From the fact that $\tuple{X\,\vec{a},Y} \in R^C$ it follows, by Definition \ref{def:can_model}, that $\assign{\vec a}\neg \varphi\in X$. Since $X$ is maximally consistent, we have that $\neg\assign{\vec a} \neg \varphi \not\in  X$, which implies, by axiom $\mathsf{N}$, that $\assign{\vec{a}}\varphi\notin X$. 
    
    Right-to-left: we again reason by contraposition. Suppose that $\tuple{X,\vec{a},Y}\notin R^C$, and thus, by the construction of the canonical model, there is a formula $\varphi \in Y$ such that $\assign{\vec a}\varphi\notin X$. Since $X$ is maximally consistent,  we have that $\neg\assign{\vec a}\varphi\in X$. By the axiom $\mathsf{N}$ we get $\assign{\vec{a}}\neg \varphi \in X$. Thus $\assign{\vec{a}}\neg \varphi \in X$ and $\neg\varphi \not\in Y$ as required for the proof. 
\end{proof}

\textbf{Proposition 7.} The canonical model $\G^C$ is a CGS.

\begin{proof}
    We have to prove that the relation $R^C$ of the canonical model is serial and functional. 
    
    For seriality, we have that given any state $X \in S^C$, 
    $X$ contains the formula $\assign{\vec{a}}\top$ for any $a\in D^\mathcal{C}$ 
    due to $\top$ being a tautology and the application of $\mathsf{Nec}$.
    Thus, by Lemma \ref{lemma:lindy}, for any $\vec{a}\in D^C$ there is a $Y\in S^C$ such that $\set{\top}\cup \set{\psi \mid  \assign{\vec a} \psi\in X} \subseteq Y $, and by Proposition \ref{prop:existforall} we have that $\tuple{X,\vec{a},Y}\in R^C$

    For functionality, suppose that $\tuple{X,\vec{a},Y}\in R^C$, $\tuple{X,\vec{a},Z}\in R^C$ and $Z\neq Y$. Thus there is a $\varphi$, such that $\varphi\in Y$ and $\neg \varphi \in Z$. By the definition of $R^C$, this implies $\assign{\vec{a}}\varphi \in X$ and $\assign{\vec a} \neg \varphi \in X$.
    By $\mathsf{N}$, the latter is equivalent to $\neg \assign{\vec a}\varphi \in X$, which contradicts the consistency of $X$.
\end{proof}

\subsection*{Model Checking}
Full model checking algorithm for $\CSL$ and the \textit{PSPACE}-hardness proof (see Algorithm \ref{cslMCfull}).
\begin{algorithm}
	\caption{An algorithm for model checking $\CSL$} \label{cslMCfull} 
	\small
	\begin{algorithmic}[1] 		
		\Procedure{MC}{$\G, s, \varphi$}		
      \Case {$\varphi = p$}
            \State{\textbf{return} $s \in \mathcal{V}(p)$}
        \EndCase
          \Case {$\varphi = \lnot p$}
            \State{\textbf{return} not $s \in \mathcal{V}(p)$}
        \EndCase
       \Case {$\varphi = \psi \lor \chi$}
            \State{\textbf{guess} $\theta \in \{\psi, \chi\}$  }
            \State{\textbf{return} $\textsc{MC} (\G,s,\theta)$}
        \EndCase
               \Case {$\varphi = \psi \land \chi$}
            \State{\textbf{universally choose} $\theta \in \{\psi, \chi\}$  }
            \State{\textbf{return} $\textsc{MC} (\G,s,\theta)$}
        \EndCase
       \Case {$\varphi = \assign{a_1,...,a_n} \psi$}
       \State{\textbf{guess} $t \in S$ such that   $\tuple{s, a_1, ..., a_n, t} \in R$}
       \State{\textbf{return} $\textsc{MC} (\G, t, \psi)$}
        \EndCase
         \Case{$\varphi = \exists x \psi$}
        \State{\textbf{guess} $a \in \Ac$  }
        \State{\textbf{return} $\textsc{MC} (\G,s,\psi[a/x])$}
        \EndCase
        \Case{$\varphi = \forall x \psi$}
        \State{\textbf{universally choose} $a \in \Ac$  }
        \State{\textbf{return} $\textsc{MC} (\G,s,\psi[a/x])$}
        \EndCase
   \EndProcedure
	\end{algorithmic}
\end{algorithm}

PSPACE-hardness is shown by the reduction from the classic satisfiability of quantified Boolean formulas (QBF), which is known to be PSPACE-complete. For a given QBF $\Psi:=Q_1 p_1 ... Q_n p_n \psi (p_1, ..., p_n)$ with $Q_i \in \{\forall, \exists\}$, the problem consists in determining whether $\Psi$ is true. Without loss of generality, we assume that in $\Psi$ each variable is quantified only once.

    Given a QBF $\Psi:=Q_1 p_1 ... Q_n p_n \psi (p_1, ..., p_n)$, we construct a CGS over one agent $\G = \langle 1, \Ac, \mathcal{D}, S, R, \mathcal{V} \rangle$, where $\Ac = \{a_1, ..., a_n\}$, $\mathcal{D} = \Ac$, $S = \{s, s_1, ..., s_n\}$, $R = \{\langle s, a_i, s_i\rangle \mid i \in \{1,...,n\}\} \cup \{\langle s_i, a_j, s_i \rangle \mid i,j \in \{1,...,n\}\}$, and $\mathcal{V}(p_i) = \{s_i\}$. Intuitively, CGS $\G$ has a starting state $s$ and a state $s_i$ for each $p_i$. The agent can reach $s_i$ from $s$ by executing action $a_i$, and all the transitions from $s_i$'s are self-loops.

    The translation from the QBF $\Psi$ into a formula $\varphi$ of $\CSL$ is done recursively as follows: 
    \begin{align*}
    \varphi_0 &:= \psi(\assign{x_1} p_1, ..., \assign{x_n}p_n) \\
    \varphi_k &:=
    \begin{cases}
        \forall x_k \varphi_{k-1} &\text{if } Q_k = \forall\\ 
        \exists x_k \varphi_{k-1} &\text{if } Q_k = \exists\\
    \end{cases}\\
\varphi &:= \varphi_n
\end{align*}
To see that 
$$Q_1 p_1 ... Q_n p_n \psi (p_1, ..., p_n) \text{ is satisfiable iff } \G,s \models \varphi$$
it is enough to notice that setting the truth-value of propositional variable $p_i$ to 1 is modelled by reachability via action $a_i$ of the state $s_i$, where $p_i$ holds. Quantifiers are modelled directly as quantifiers over the agent's actions. 

As an example, consider a QBF $\forall p_2 \exists p_1 \exists p_3 (p_1 \to p_2) \land p_3$, which is clearly satisfiable with $p_1 = 0$ and $p_3 = 1$. The formula is translated into the formula of $\CSL$: $\forall x_2 \exists x_1  \exists x_3 (\assign{x_1}p_1 \to \assign{x_2}p_2) \land \assign{x_3}p_3$. The corresponding CGS is presented in Figure \ref{fig::hardness}, and it is easy to verify that $\G,s \models \forall x_2 \exists x_1  \exists x_3 (\assign{x_1}p_1 \to \assign{x_2}p_2) \land \assign{x_3}p_3$.

\begin{figure}[h!]
\centering
\scalebox{0.9}{
\begin{tikzpicture}
\node(-1) at (2,0) {};
\node[circle,draw=black, minimum size=4pt,inner sep=0pt,  label=left:{$s$}](1) at (0,0) {};
\node[circle,draw=black, minimum size=4pt,inner sep=0pt, , label=below:{$s_1$}, label = right:{$\{p_1\}$}](2) at (-2,-2) {};
\node[circle,draw=black, minimum size=4pt,inner sep=0pt, , label=below:{$s_2$}, label = right:{$\{p_2\}$}](3) at (0,-2) {};
\node[circle,draw=black, minimum size=4pt,inner sep=0pt, , label=below:{$s_3$}, label = right:{$\{p_3\}$}](4) at (2,-2) {};

\draw [->,thick](1) to node[left,align=left] {$a_1$} (2);
\draw [->,thick](1) to node[left,align=left] {$a_2$} (3);
\draw [->,thick](1) to node[right,align=left] {$a_3$} (4);
\draw [->,thick] (2) to [loop left]  (2);
\draw [->,thick] (3) to [loop left]  (3);
\draw [->,thick] (4) to [loop left]  (4);
\end{tikzpicture}
}
\caption{CGS $\G$ for a single agent. Labels for self-loops are omitted for readability.}
\label{fig::hardness}
\end{figure}
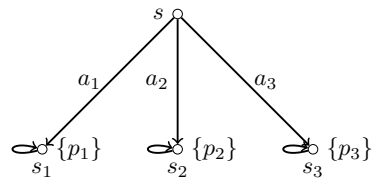 

\subsection*{Satisfiability}
The proof of undecidability of $\mathsf{SL}$ \cite{mogavero16} uses only the next-time fragment of $\mathsf{SL}$. Here we show how the proof can be adapted for the case of $\CSL$. For more details see the original proof in \cite{mogavero16}.

Given two variables $x_1$ and $x_2$, let $x_1 < x_2$ be the formula $\assign{x_1,y} p \land \assign{x_2,y}\neg p$. Define also $\varphi_{unbd}= \forall x_1 \exists x_2 (x_1 < x_2)$, $\varphi_{mn}= \exists x_2 \forall x_1 \neg (x_1 < x_2)$, and 
$\varphi_{trs}= \forall x_1 \forall x_2 \forall x_3 ((x_1 < x_2 ) \land  (x_2 < x_3 ))\imp (x_1 < x_3) $. Finally let 
    $\varphi_{<} =\varphi_{unbd}\land \varphi_{mn} \land \varphi_{trs}$. 

   \begin{proposition}
       Formula $\varphi_<$ is satisfiable.
   \end{proposition}
   \begin{proof}
       Consider the $CGS$ $\G^\star$ with three states $s_0,s_1,s_3$, where $\Ac=\mathbb{N}$, and where we have that  $\tuple{s_0,k,j,s_1}\in R$ iff $k < j$, $\tuple{s_0,k,j,s_2}\in R$ iff $k\geq j$ and $\tuple{s_i,k,j,s_i}$ for $i=1,2$ and $k,j\in \mathbb{N}$. Moreover suppose that $\mathcal{V}(s_0)=\mathcal{V}(s_2)=\emptyset$ and $\mathcal{V}(s_1)=\set{p}$. It is easy to verify that $\G^\star,s_0 \models \varphi_<$. 
   \end{proof}

   In the proof of the fact that $\varphi_<$ is satisfiable, we constructed a model over an infinite number of actions. We can now show that we cannot do with less, i.e. every model of $\varphi_<$ necessarily possesses an infinite number of actions. 

   \begin{lemma}
       Let $\G$ be a $CGS$ whose set of actions is $\Ac$, and $s$ is one of its states. Define a relation $\prec \subseteq \Ac \times \Ac$ by $\tuple{a,b}\in \prec $ iff $\G,s\models x_1 < x_2[a/x_1, b/x_2]$. If $\G,s\models \varphi_{<}$, then relation $\prec$ is a strict partial order without the maximal element, and hence $\Ac$ is infinite. 
   \end{lemma}
   \begin{proof}
     That $\prec$ is transitive and unbounded immediately follows from the fact that $\G,s\models \varphi_<$. The fact that no maximal element exists for $\prec$ follows from the irreflexivity of $\prec$ 
     Indeed, suppose that $\tuple{a,a}\in \prec$ for some action $a$. This means that $\G,s\models \assign{a,b} p \land \assign{a,b}\neg p$ for some $b$, i.e., for the unique $s'$ s.t. $\tuple{s,a,b,s'}\in R$ we have  that $\G,s'\models p\ \land \neg p$, which is a contradiction. 
     Since $\prec$ is a strict partial order without maximal elements on $\Ac \times \Ac$, we directly obtain that $\Ac$ must be infinite. 
   \end{proof}

Define $a\equiv b$ iff neither $a\prec b$ nor $b\prec a$. It is easy to see that if $\G,s\models \varphi_<$, then $\equiv$ is an equivalence relation on the set of actions of $\G$. Let us denote by $[a]$ the equivalence class of $a\in \Ac$ generated by $\equiv$. We define a  relation $\prec^\equiv$ on equivalence classes of actions by putting $[a]\prec^\equiv [b]$ iff for every $a\in [a]$ and every $b\in [b]$,  $a\prec b$ holds. We can now show the following property of $\prec^\equiv$. 

\begin{lemma}
   Suppose that $\G,s\models \varphi_<$, then the relation $\prec^\equiv$ on the class of equivalence classes generated by $\equiv$ is a strict total order with the minimal element and no maximal element. 
\end{lemma}
\begin{proof}
    Transitivity, irreflexivity and the non-existence of a maximal element are immediate.   We show that $\prec^\equiv$ is total and admits a minimal element. For totality, we argue by contradiction: assume that $[a]$ and $[b]$ are two different incomparable equivalent classes. By definition, this means that there are $a,a'\in [a]$ and $b,b'\in [b]$ such that $a \not \prec b$ and $b'\not \prec a'$. This means that $b\prec a$ and $a' \prec b'$ must hold. Since $a$ and $b'$ are in different equivalence classes, 
    w.l.o.g. we can assume then that $a\prec b'$, which implies $b\prec b'$ by transitivity. But this is impossible since $b,b'\in [b]$. 
    
    Now suppose that there is no minimal element for $\prec^\equiv$. This implies that in $\G$ for every action $a \in [a]$ we can find an action $b \in [b]$ such that  such that $\G,s\models x_1 < x_2 [b/x_1, a/x_2]$ and thus $\G,s\models \forall x_2 \exists x_1 (x_1 < x_2)$. But this contradicts the conjunct $\varphi_{mn}$ of $\varphi_<$, i.e. $\G,s\models \exists x_2 \forall x_1 \neg (x_1 < x_2)$. 
\end{proof}
   
Having defined appropriate relations, we can use them to capture the construction of a tiling.
 Given a finite set $D$ of domino types and two relation $H,V\subseteq  D\times D$ the $\mathbb{N}\times \mathbb{N}$, \textit{the domino tiling problem} consist in finding a mapping $T: \mathbb{N}\times \mathbb{N} \to D$ such that for every $x,y\in \mathbb{N}$ it holds that $T(x,y)\in H$ implies $T(x+1,y)\in H$, and $T(x,y)\in V$ implies $T(x,y+1)\in V$. We can now use the reduction from the tiling problem to show the undecidability of the satisfiability problem for $\CSL$.

 Given $x_1$ and $x_2$, we define the two formulae $x_1 <_H x_2 = \exists y \assign{x_1,y} p \land \assign{x_2,y} \neg p$  and $x_1<_V x_2= \exists y \assign{y,x_1} p \land \assign{y,x_2}$. For $X\in \set{H,V}$, let $\varphi^X_{<}$ be defined symilarly to $\varphi_<$.
 Finally, let $\varphi_{grd}$ be the conjunction of $\varphi^H_<$ and $\varphi^V_<$. Intuitively, $\varphi^{grd}$ enforces the horizontal and vertical orderings of the positions in a grid.  Evidently, the sentence $\varphi^{grd}$ is satisfiable and every of its model has a countable number of actions. 
 
 By defining, for $X\in \set{H,V},$ $a\prec_X b$  similarly to $a\prec b$, we obtain that $\prec_X$ is a strict partial order with no maximal element on the set of actions on every model $\G$ of $\varphi^{grd}$. Let $\equiv_X$ for $X\in \set{H,V}$ be the equivalence relation on actions of $\G$ generated as $\tuple{a,a'}\in \equiv_X$ iff neither $a \prec_X a'$ nor $a'\prec_X a$. Then, by denoting $\Ac^\equiv_X$ the class of equivalence classes modulo $\equiv_X$, we get that $\prec^\equiv_X$ is a strict total order with the minimal element and no maximal element. Since $\Ac^\equiv_X$ is countable, it follows that $\tuple{\Ac^\equiv_X,\prec^\equiv_X}$ is a well-ordered set. We denote by $[a]^X_i$ the $i$-th element of $\Ac^\equiv_X$ w.r.t. $\prec^\equiv_X$. 
 
 Let $\mathcal{N}: \mathbb{N}\times \mathbb{N} \to \Ac^\equiv_H \times \Ac^\equiv_V$ be a mapping that maps any pair $\tuple{i,j}$ of natural numbers to the pair $\tuple{[a]^H_i, [a]^V_j}$.  
 We can define the successor relation 
 $S^X(x_1,x_2)$ as $(x_1 \prec_X x_2) \land \forall x_3 \neg((x_3 \prec_X x_2) \land (x_1 \prec_X x_3))$
 for $X\in \set{H,V}$. In such a way we ensure that if $\G,s\models \varphi_{grd}$ then $\G,s\models S^X(x_1,x_2)[a/x_1,a'/x_2]$ iff $a\in [a]^X_i$ and $a'\in [a]^X_{i+1}$. 
 
 Having defined the successor relation, we can now express the local compatibility of a tiling  $\varphi^{loc,t}$, i.e. that each tile has only one type, as well as horizontal and vertical requirements of a tiling $\varphi^{t,H}$ and $\varphi^{t,V}$.

 \begin{enumerate}
     \item $\varphi^{loc,t}= \assign{x,y}(t\land^{t'\neq t}_{t'\in D} \neg t')$
     \item $\varphi^{t,H}= \bigvee_{\tuple{t,t'}\in H} (\forall x (S^H (x,x') \imp \assign{x',y} t'))$
     \item $\varphi^{t,V}= \bigvee_{\tuple{t,t'}\in V} ( \forall y (S^V (y,y') \imp \assign{x,y'} t'))$
 \end{enumerate}

 Finally, let $\varphi^{tile}$ be $\forall x \forall y ( \varphi^{loc,t} \land \varphi^{t,H} \land \varphi^{t,V})$, and let the formula for the tiling problem be $\varphi^{dom}:= \varphi^{grd} \land \varphi^{tile}$. 

\begin{theorem} The satisfiability problem for $\CSL$ is undecidable. 
\begin{proof}
    To obtain the result one shows that there is a reduction from  the $\mathbb{N}\times \mathbb{N}$ tiling problem to the satisfiability problem of $\CSL$. Having defined all the necessary formulae above, the proof goes similarly to \cite[Theorem 3.10]{mogavero16}.
\end{proof}
    
\end{theorem}

\end{document}